\documentclass[11pt,reqno]{amsart}
\usepackage{amsfonts,amssymb,amsmath,amsopn,amsthm,graphicx}
\usepackage{amsxtra, mathrsfs}
\usepackage{colordvi}
\usepackage[usenames,dvipsnames]{color}
\usepackage{amsfonts,amssymb,amsbsy,amsmath,amsthm,dsfont}
\usepackage{mathtools}
\usepackage[colorlinks=true]{hyperref}
\usepackage{wasysym}
\usepackage{enumitem}

\usepackage[parfill]{parskip} 

\usepackage[parfill]{parskip} 

\numberwithin{equation}{section}

\newtheorem{theorem}{Theorem}[section]
\newtheorem{corollary}[theorem]{Corollary}
\newtheorem{lemma}[theorem]{Lemma}
\newtheorem{proposition}[theorem]{Proposition}

\theoremstyle{definition}

\newtheorem{remark}[theorem]{Remark}

\newtheorem{assumption}[theorem]{Assumption}

\newcommand{\ball}{\mathscr{B}}

\renewcommand{\epsilon}{\varepsilon}

\newcommand{\loc}{{\rm loc}}

\renewcommand{\phi}{\varphi}
\newcommand{\R}{\mathbb{R}}

\newcommand{\Sph}{\mathbb{S}}

\newcommand{\w}{\mathrm{w}}
\newcommand{\Z}{\mathbb{Z}}
\newcommand{\eps}{\epsilon}

\newcommand{\V}{\mathcal{V}}

\newcommand{\mynegspace}{\hspace{-0.12em}}
\newcommand{\lvvvert}{\rvert\mynegspace\rvert\mynegspace\rvert}
\newcommand{\rvvvert}{\rvert\mynegspace\rvert\mynegspace\rvert}
\DeclarePairedDelimiter{\vvvert}{\lvvvert}{\rvvvert}

\newcommand{\displaystyleint}{\displaystyle\int}

\author[H. Kova\v{r}\'{\i}k]{Hynek Kova\v{r}\'{\i}k}
\author[P. C. Rossaro]{Pier Cristoforo Rossaro}

\address[Hynek Kova\v{r}\'{\i}k]{
DICATAM, Sezione di Matematica
Universit\'a degli studi di Brescia
Via Branze, 38 - 25123
Brescia, ITALY}
\email{\href{mailto:hynek.kovarik@unibs.it}{hynek.kovarik@unibs.it}}

\address[Pier Cristoforo Rossaro]{
DICATAM, Sezione di Matematica
Universit\'a degli studi di Brescia
Via Branze, 38 - 25123
Brescia, ITALY}
\email{\href{mailto:p.rossaro@unibs.it}{p.rossaro@unibs.it}}

\keywords{Magnetic Schr\"odinger operator, Hardy Inequality}

\makeatletter
\@namedef{subjclassname@2020}{%
  \textup{2020} Mathematics Subject Classification}
\makeatother

\subjclass[2020]{Primary: 35R45; Secondary: 47A63}

\begin{document}
\title[Magnetic Hardy  inequalities with singular integral weights]{Magnetic Hardy  inequalities with singular integral weights}

\maketitle

\begin{abstract}
In this  paper we present  Hardy type inequalities for magnetic Dirichlet forms with singular integral weights. We analyze the local 
and global optimality of the integral weight and discuss several examples in details. An application of our results to 
spectral estimates for magnetic Schr\"odinger operators is provided as well. 
\end{abstract}

\section{\bf Introduction}

\subsection{\bf The background}
The classical Hardy inequality in $\R^n$ states that 
\begin{equation} \label{clas-hardy} 
\int _{\R^n}  |\nabla u|^2 \, dx  \, \geq \, \frac{(n-2)^2}{4} \int_{\R^n}  \frac{ |u|^2}{|x|^2}\, dx ,\qquad n\geq 3
\end{equation} 
for all $u\in H^1(\R^n)$. Inequality \eqref{clas-hardy} is a multidimensional version of the original 
one-dimensional bound proved by Hardy in 1920, \cite{hardy}. 
It is well known that the constant on the right hand side of \eqref{clas-hardy} 
is sharp and that the inequality fails, for any positive integral weight, in dimension one and two. For further reading we refer to the monographs \cite{HLP, ok} and references therein.

Hardy inequality for magnetic Dirichlet forms has much shorter history than its non-magnetic counterpart \eqref{clas-hardy}. Nevertheless,  starting from the pioneering 
work by Laptev and Weidl \cite{lw-1999} it has produced a considerable amount of literature, see e.g.~\cite{w-cpde-1999, w-jlms-1999, bls, cfkp-2021, ck, ckll, ek, ep, et, fk-2024, fklv} and the monograph \cite{bel}. The essential contribution of  \cite{lw-1999} is the proof of the fact that introducing a magnetic field $B:\R^2\to \R$, which means replacing $\nabla$ in \eqref{clas-hardy} by a magnetic gradient $i\nabla +A$ with $A:\R^2\to\R^2$ such that $\nabla\times A =B$, leads to a Hardy type inequality even in dimension two; 
\begin{equation} \label{generic} 
\int_{\R^2} |(i\nabla+A)u|^2 \, dx \, \geq \, \int_{\R^2} \rho\, |u|^2 \, dx\, .
\end{equation}
Such a bound holds under suitable regularity conditions on $B$ and 
for test functions $u$ in the corresponding form domain. By \cite{lw-1999, w-cpde-1999}, 
 we have  $ \rho\ \gneqq 0$ as soon as $B\neq 0$. It should be pointed out however that while generically
$ \rho(x) \sim\, |x|^{-2}$ at infinity, in a neighbourhood of zero one has  $ \rho(x) \sim |x|^{-2}$ only in the special case of the so-called 
Aharonov-Bohm magnetic field which is concentrated in one point. In fact, in that case we have
\begin{equation} \label{ab}
 \rho(x) = \frac{\min_{k\in\Z} |\mu-k|^2}{|x|^2} \qquad \text{if} \qquad A(x) = A_{ab}(x)= \mu\, \frac{(-x_2, x_1)}{|x|^2} \, , \\[3pt]
\end{equation} 
where the constant $\min_{k\in\Z} |\mu-k|^2$ is sharp, see \cite{lw-1999}. On the other hand, if the magnetic field is sufficiently regular so that $|A|\in L^q_{\rm loc}(\R^2)$ with 
$q>2$, then $\|(i\nabla +A) u\|_2^2 +\|u\|_2^2$ is {\it locally} equivalent to $\|u\|_{H^1}^2$, cf.~Theorem \ref{thm-weidl}. Hence the domain of the quadratic form in \eqref{generic} includes also functions which are constant in the vicinity of zero, and therefore the weight function  $ \rho$ cannot be as singular as in \eqref{clas-hardy} and \eqref{ab}. To see this observe that  the function  $|\cdot|^{-2}$ is not locally integrable in $\R^2$.

Yet the validity of \eqref{generic} for a weight function $\rho$ with an integrable singularity in the origin cannot be excluded  a priori. Nonetheless, in all the above mentioned papers concerning dimension two, with the exception of \cite{cfkp-2021}, the integral weight $ \rho$ is {\it locally bounded}. 

The main object of interest of  \cite{cfkp-2021} are magnetic operators 
on  Heisenberg group. However, in \cite[Rem.~6.6]{cfkp-2021} the authors observe that their method can be applied also to magnetic Schr\"odinger operators to show that if $B$ is smooth, compactly supported and such that $\frac{1}{2\pi} \int_{\R^2} B \not\in \Z$, then inequality \eqref{generic} holds with
\begin{equation} \label{cfkp}
 \rho(x) \, = \frac{c(B,r_0)}{r^2 \big[1 + ( \log_+ r_0/r)^2\big] }\, ,  \qquad r = |x|. 
\end{equation} 
Here $r_0$ is large enough so that supp$B\subset \{x\in\R^2 : \ |x| < r_0\}$, and $\log_+$ denotes the positive part of  the logarithm. The constant in \eqref{cfkp} depends only on the magnetic field and on the choice of $r_0$. This result was later extended to magnetic $p-$Laplace operator in \cite{ckll}.

In this paper, inspired by \cite{cfkp-2021}, we address the following questions. First, what is the strongest possible singularity of $ \rho$ in the Hardy inequality \eqref{generic} 
generated by a  regular magnetic field? Here and below  we call a magnetic field $B:\R^2\to\R$  {\it regular} if 
\begin{equation} \label{ass-regular}
B \in L^p_{\rm\, loc} (\R^2) \quad \text{for some}\quad p>1. 
\end{equation}
Second, what is the relation between the local behavior of $\rho$ and $B$ in a situation in which the field is not regular but still belongs to $L^1_{\rm loc}(\R^2)$? Recall that 
that the Aharonov-Bohm magnetic field generated by the vector $A_{ab}$ given in \eqref{ab} is not locally integrable since in that case we have $B(x)=0$ for any $x\neq 0$, while $\int_{\R^2}B\neq 0$.

As for the first question, it turns out that the weight \eqref{cfkp} found in \cite{cfkp-2021} is the best possible, at least as far as the logarithmic scale is concerned, for {\it all regular} magnetic fields, see Theorem \ref{thm-hardy} and Remark \ref{rem-1}. Moreover, the decay rate of $\rho_0$ is optimal also at infinity if no further restrictions on $B$ are introduced, see Remark \ref{rem-2} and Proposition \ref{prop-log-global}.

To answer the second question we consider, for definiteness, magnetic fields $B\in L^1_{\rm loc}(\R^2)$ with a single isolated singularity in the origin. The behavior of $\rho$ in a neighborhood of zero is then characterized 
 in terms of the (normalized) flux of $B$ through the ball of radius $r$ centered in the origin: 
\begin{equation} \label{alpfa-r} 
 \beta(r) = \frac{1}{2\pi} \int_{|x|<r} B(x) \, dx .
\end{equation}
Accordingly we define 
\begin{equation} \label{Phi} 
\Phi(r) = \frac{ \beta(r)}{r}\, .
\end{equation}
Our main result concerning singular fields then implies that if
the singularity of $B$ is strong enough, then  $\rho(x) \sim\Phi(r)^2$ as $r\to 0$, see Theorem \ref{thm-singular-local} for more details. We also characterize the domain on which the quadratic form $\|(i\nabla+A)u\|_2^2$ is closed, and show how it depends on the asymptotic behavior of the magnetic field around zero. This is done in Corollary \ref{cor-domain}.

\subsubsection*{\bf Notation}
Let $X\subseteq \R^2$ be an arbitrary set and $f,g:X \to \mathbb{R}$. In the following, we write 
$$
f(x) \lesssim  g(x)
$$ 
if
there exists a constant $c>0$,   such that $f(x) \leq c \, g(x)$ for all $x\in X$. The symbol $\gtrsim$ is defined  accordingly, We  write $f(x) \asymp g(x) $ if $f(x) \gtrsim g(x)$ and $f(x) \lesssim g(x)$. 

By $\ball_R=\{x\in\R^2: \, |x|< R\}$ we denote the ball of radius $R$ centered at the origin. We also introduce the function $\log_+: \R_+\to \R_+$ defined by
\begin{equation} \label{log+}
\log_+(x) = \max\{0, \log (x)\}, \qquad x>0.
\end{equation}
Finally, by $\Sph$ we denote the unit sphere in $\R^2$.

\subsection{\bf Main results} Throughout the paper we denote by
 \begin{equation} \label{eq-form}
Q_A[u] = \int_{\R^2} |(i\nabla+A)u|^2 \, dx 
\end{equation} 
the quadratic form associated to $A$. The associated  graph norm is given by
\begin{equation} \label{vvv}
\vvvert{u}_A = \big( Q_A[u] + \|u\|_2^2 \big)^{1/2}\, . 
\end{equation}
Let 
\begin{equation}
d_0(Q_A) =  \overline{C_0^\infty(\R^2)}^{\, \vvvert{ \cdot}_A}\, .
\end{equation}

 Our first result shows that the singularity in \eqref{cfkp} is generic for all  regular magnetic fields.

\begin{theorem}[\bf Hardy inequality for regular magnetic fields] \label{thm-hardy} 
Let $B\neq 0$ satisfy  \eqref{ass-regular}, and let $A:\R^2\to\R^2$ be a vector field such that $\nabla \times A =B$. Then there exists
a constant $C_B$ such that for all $u\in d_0(Q_A)$ we have 
\begin{equation} \label{hardy-eq}
Q_A[u]  \, \ge\,  C_B \int_{\R^2} \rho_0\, |u|^2 dx
\end{equation}
where
\begin{equation} \label{weight-rho}
\rho_0(x) = \frac{1}{r^2 (1 + \log^2 r)}\, .
\end{equation}
\end{theorem}

\smallskip

A couple of comments concerning Theorem \ref{thm-hardy} are in order

\begin{remark}[\bf Optimality at zero]\label{rem-1} 
The power of the logarithm on the right hand side of \eqref{hardy-eq} is optimal locally in the vicinity of zero in the 
following sense. Given any magnetic field satisfying \eqref{ass-regular} the inequality 
\begin{equation} \label{eq:opt_2}
\displaystyleint_{\ball_{1/e}} \dfrac{|u|^2}{r^2\, |\log r|^b} \, dx\  \le\  c\,  Q_A[u] \qquad u\in d_0(Q_A)
\end{equation}
fails for any $c\in\R$ and any $b<2$. To see this consider the radial function $u_ \beta(r) = |\log_+(1/ r)|^\alpha$.  Clearly $u_\alpha$ is compactly supported and a quick calculation shows  that $u_\alpha\in H^1_{\rm loc}(\R^2)$ for {\it any} $\alpha < \frac 12$.
Hence, by Theorem \ref{thm-weidl}(1), we have $|(i\nabla +A) u_\alpha| \in L^2(\R^2)$. On the other hand, the left hand side of  
\eqref{eq:opt_2} with $u=u_\alpha$ is finite if and only if $b> 1 +2\alpha$.
\end{remark}

\smallskip

\begin{remark}[\bf Optimality at infinity]\label{rem-2} Without further restrictions on the magnetic field, the power of the logarithm on the right hand side of \eqref{hardy-eq} is optimal also at infinity. 
This follows from Proposition \ref{prop-log-global} below. Note however, that for magnetic fields with finite non-integer flux the logarithmic factor in \eqref{hardy-eq} can be removed, see \cite{lw-1999}.
\end{remark}

\begin{remark}
In the last section of the paper we discuss an application of inequality \eqref{hardy-eq} to spectral inequalities for two-dimensional magnetic Schr\"odinger operators, see Theorem \ref{thm-counting}. 
\end{remark}

In order to state our main result for singular fields we need  to introduce some additional notation. We say that a magnetic field $B$ is  {\it singular} if it can be decomposed into 
a sum 
\begin{equation} \label{B-decomp}
  B =  B_s + B_r, 
\end{equation}
where $B_r$ satisfies \eqref{ass-regular} and $B_s\in L^1_{\rm loc}(\R^2)$  is has an isolated singularity in the origin which is strong enough so that 
\begin{equation} \label{bs-diverge}
\forall\, \eps>0: \quad \lim_{r\to 0} r^{-\eps}\!\! \int_{\ball_r} B_s\, dx = \pm \infty. 
\end{equation}
Notice that this condition contradicts the validity of \eqref{ass-regular}  for $B$.

\begin{assumption}\label{ass-singular}
Suppose that $B$ satisfies \eqref{B-decomp} and \eqref{bs-diverge} with $B_s$ being compactly supported and  bounded on compact
subsets of $\R^2\setminus\{0\}$. 
 \end{assumption}
 
 \begin{remark}
Since $B_s$ is bounded away from zero, the assumption that the support of $B_s$ is compact represents no loss of generality. Indeed, it can always be achieved by a suitable modification of the regular component $B_r$. 
 \end{remark}

The domain of the quadratic form $Q_A$ is then given by
\begin{equation} \label{domain-s}
d(Q_A) = \overline{C_0^\infty(\R^2\setminus\{0\})}^{\, \vvvert{ \cdot}_A}\, .
\end{equation}

\medskip

In the next result we quantify the local contribution of the flux to the integral weight in the Hardy inequality.

\begin{theorem}[\bf Hardy inequality for singular magnetic fields] \label{thm-singular-local} 
Let $B $ satisfy Assumption \ref{ass-singular}. Then there exists $\eta>0$ and a constant $c(B,\eta)$ such that
 for all $u\in C_0^\infty(\R^2\setminus\{0\})$ we have
\begin{equation}  \label{hardy-singular-local}
Q_A[u] \, \geq\, c(B,\eta) \int_{\R^2} \rho\, |u|^2\, dx\, ,
\end{equation} 
where
\begin{equation}\label{rho}
\rho(r) = \begin{cases}
\rho_0(r) + \Phi^2(r)  & \text{ if} \quad r \leq \eta \\[3pt]
 \rho_0(r)  & \text{ if} \quad  r > \eta
	\end{cases}\, .
\end{equation}
Recall that $\beta(\cdot)$ is defined in \eqref{alpfa-r} and that $\Phi$ is defined in \eqref{Phi}. 
\end{theorem}

\begin{remark} 
It follows from  equation \eqref{alpfa-r} that the term $\Phi^2(r)$ can be arbitrarily close to the  singularity generated by the Aharonov-Bohm vector potential $A_{ab}$, cf.~equation \eqref{ab}. Hence the weight function covers, by varying $\beta(\cdot)$, the whole ``gap" between $\rho_0(\cdot)$ and $|\cdot|^{-2}$.  
Explicit  examples are given in Section \ref{sec-examples}, see in particular Example \eqref{ex-1}.  \end{remark}
\smallskip 

\begin{remark} 
Notice also that in the case of the Aharonov-Bohm field, which is not covered by Theorem \ref{thm-singular-local}, we have $ \beta(r) =${\it constant}, see \eqref{ab}.
\end{remark}

\smallskip 

\begin{remark} 
A comparison of equations \eqref{rho} and \eqref{weight-rho} shows that the contribution from $\Phi$ becomes dominant  as soon as $ \beta(r)\to 0$ slowly enough so that 
$| \beta(r) \log r | \to \infty$ as $r \to 0$, see Corollary \ref{cor-domain}.
\end{remark}

\subsubsection*{\bf Outline of the paper} Theorems \ref{thm-hardy} and \ref{thm-singular-local} are proved in Sections \ref{sec-regular} and \ref{sec-singular} respectively. In the closing Section \ref{sec-application} we then discuss an application of inequality \eqref{hardy-eq} to non-semiclassical estimates for the number of negative eigenvalues of magnetic Schr\"odinger operators.

\section{\bf Regular magnetic fields}
\label{sec-regular}

We start by describing our choice of the gauge. Given a regular magnetic field $B$ we consider a distributional solution of the equation 
\begin{equation} \label{h-B}
-\Delta h = B\, .
\end{equation}
Then  the vector field $A: \R^2 \to \R^2$ given by
\begin{equation} \label{A-def}
A = (\partial_{x_2}h, -\partial_{x_1}h)
\end{equation}
satisfies $\nabla \times A =B$ in the sense of distributions. 

\begin{lemma} \label{lem-A-regular}
Suppose that $B$ satisfies \eqref{ass-regular} and let $A$ be given by \eqref{A-def}. Then $|A|\in L^{q}_{\loc}(\R^2)$  for some $q>2$.
\end{lemma}

\begin{proof} 
Without loss of generality we may assume that $B \in L^p_{\rm\, loc} (\R^2)$ with $1<p<2$.  By the elliptic regularity we then have $h \in W_{\loc}^{2,p}(\R^2)$,
see e.g.~ \cite[Thm.~9.15]{GT}. Now the Sobolev imbedding, see e.g.~\cite[Thm.~8.8]{lieb-loss},  implies that  $|\nabla h| \in L^q_{\loc}(\R^2)$ for any 
$q\leq \frac{2p}{2-p}$. Since  $\frac{2p}{2-p} >2$ by assumption, this proves  the claim.
\end{proof}

In view of Lemma \ref{lem-A-regular} and  \cite[Theorem 2.2]{s1}  it follows that the quadratic form \eqref{eq-form}
is closed on the domain
\begin{equation} \label{q-form-1}
d_0(Q_A) = \big\{ u\in L^2(\R^2): \  |(i\nabla+A)u|\in L^2(\R^2)\big\}.
\end{equation}
where $\vvvert{u}^2_A$ is defined in \eqref{vvv}.


The following theorem, which is due to Weidl, will be important in the proof of our main results.

\begin{theorem}[Weidl] \label{thm-weidl}
Suppose that  $|A|\in L^{q}_{\loc}(\R^2)$  for some $q>2$. 
\begin{enumerate} 
\item  Let $U\subset\R^2$ be a bounded connected Lipschitz domain and let $\Omega\subseteq U$ be a set of positive Lebesgue  
measure. Then
for any $u\in H^1(U)$, 
\begin{equation} \label{equiv-1}
\|(i\nabla+A)u\|_{L^2(U)}^2+ \|u\|^2_{L^2(\Omega)} \ \asymp\ \|u\|_{H^1(U)}^2\, .
\end{equation}

\item If $B\neq 0$, then for any compact set $\Omega\subset\R^2$ of positive Lebesgue  measure there exists a constant $c(\Omega,B)$ such that 
\begin{equation} 
Q_A[u] \ge c(\Omega,B) \int_\Omega |u|^2  \, dx 
\end{equation}
for all $u\in d(Q_A)$.
\end{enumerate}
\end{theorem}
 
For the proof of Theorem \ref{thm-weidl} we refer to \cite[Cor.~9.2]{w-cpde-1999} and  \cite[Lemma 4.1]{w-jlms-1999}. 

\medskip

We now  prove a result which plays a crucial role in the proof of inequality \eqref{hardy-eq}.  
In order to simplify the notation in the sequel, given a bounded set $\Omega$ we introduce the shorthand 
 \begin{equation}\label{w-weight}
    \|u\|_{L^2_w(\Omega)}=\Big\| \dfrac{u}{r \log(r/r_0)} \Big\|_{L^2(\Omega)}\, ,
 \end{equation}
where $r_0$ is large enough so that $\Omega \Subset \ball_{r_0}$.

\begin{lemma} \label{lem-equiv-2}
 Let $\Omega \subset\R^2$ be a compact set with $C^1$ boundary and 
 suppose that $\Omega$ contains an open neighbourhood of the origin. 
 Let $r_0$ be such that $\Omega \Subset \ball_{r_0}$. Then, for any $u\in H^1(\Omega)$, 
\begin{equation}\label{eq: equiv_norm_1}
 \|u\|^2_{L^2_w(\Omega)} \ \leq \ C\, \|u\|_{H^1(\Omega)}^2\, .
\end{equation}
where the constant $C$ depends only on $\Omega$. 
\end{lemma}

Note that  $w \notin L^p_{\loc}(\R^2)$ for any $p>1$.  Hence \eqref{eq: equiv_norm_1}
doesn't follow from standard imbedding theorems.

\begin{proof}
Let $u\in\ H^1(\Omega)$. 
We consider the integral
\begin{equation}
\displaystyleint_{\Omega} \big|\nabla u + \dfrac{f u}{r} x\big|^2 dx \, \geq\, 0,
\end{equation}
where $f $ is a radial function to be determined later. Then
\begin{align*}
\displaystyleint_{\Omega} \big|\nabla u + \dfrac{f u}{r} x\big |^2 dx &= \displaystyleint_{\Omega} \left(|\nabla u|^2 +  \dfrac{f}{r} \nabla(u^2) \cdot x + f^2 u^2\right) dx,
\end{align*}
and  an integration by parts implies
\begin{equation*}
\displaystyleint_{\Omega} \dfrac{f}{r} \nabla(u^2) \cdot x dx = \int_{\partial\Omega} u^2 \dfrac{f}{r} x \cdot \nu d\sigma - \displaystyleint_{\Omega} u^2 \nabla \cdot \left(\dfrac{x f}{r}\right) dx\, .
\end{equation*}
Here  $\nu$ denotes the external normal unit vector to the boundar Hence
\begin{align}\label{eq:quad_for_Lem_9.1}
0\leq \displaystyleint_{\Omega} \big|\nabla u + \dfrac{f u}{r} x\big |^2 dx &= \displaystyleint_{\Omega} \Big(|\nabla u|^2  -u^2  \big(f' + \dfrac{f}{r} - f^2\big )\Big)dx  +  BT(u),
\end{align}
where we have denoted the boundary term by
\begin{equation}\label{bt}
BT(u):=\int_{\partial\Omega} u^2\, \dfrac{f}{r} \, x \cdot \nu d\sigma\, .
\end{equation}
From(\ref{eq:quad_for_Lem_9.1})  we obtain
\begin{equation} \label{2-upperb}
\displaystyleint_{\Omega} u^2 \big(f' + \dfrac{f}{r} - f^2\big )\, dx  \le \displaystyleint_{\Omega} |\nabla u|^2 + BT(u).
\end{equation}

Now we choose 
\begin{equation} \label{f-def} 
f = -\dfrac{1}{2 r\log(r/r_0)}
\end{equation}
so that
\begin{equation} \label{weight-log}
f' + \dfrac{f}{r} - f^2 = \dfrac{1}{4r^2 (\log(r/r_0))^2}\, .\\[2pt]
\end{equation}
Next we treat the boundary term \eqref{bt}.
Since $0\not\in\partial\Omega$ by assumptions, we have
$|f|_{\partial\Omega} \lesssim 1 $. Hence by the boundary trace theorem, see e.g.~\cite[Theorem 5.22]{adams-ss}, 
there exists a constant $C=C(\Omega)$ such that
\begin{equation}
BT(u) \ \leq C(\Omega)  \displaystyleint_{\partial\Omega} u^2 d\sigma\ \lesssim\  \|u\|^2_{H^1(\Omega)}\, .
\end{equation}
Inserting this bound together with \eqref{weight-log} into \eqref{2-upperb}  proves  \eqref{eq: equiv_norm_1}
\end{proof}

\begin{remark}
The fact that $    \|u\|_{L^2_w(\Omega)} < \infty $ for any $u\in H^1_0(\Omega)$ is known, see e.g.~\cite{bw}, \cite{tar}, \cite{iok}.
\end{remark}



\begin{proof}[\bf Proof of Theorem \ref{thm-hardy}]
Owing to the gauge invariance of inequality \eqref{hardy-eq} we may work with the vector $A$ given by \eqref{A-def}.
By Lemma \ref{lem-A-regular} and \cite[Thm.~2,2]{s1}  it follows that  $C_0^\infty(\R^2)$ is dense in $d(Q_A)$ in 
the norm \eqref{vvv}.
Hence it suffices to prove  \eqref{hardy-singular-local}  for any $u\in C_0^\infty(\R^2)$. For then, for any $\eps>0$, we have  \eqref{hardy-singular-local}, on  $C_0^\infty(\R^2)$,  with the 
bounded  integral weight $\big(\eps+ r^2 (1 + \log^2 r)\big)^{-1}$. By density of  $C_0^\infty(\R^2)$ in $d(Q_A)$, this extends to all $u\in d(Q_A)$. Letting $\eps\to 0$ and using the monotone convergence then gives the claim.

We may suppose that  $B$ is not identically zero in $\R^2$. Then by  Theorem \ref{thm-weidl}(2) 
\begin{equation} \label{hardy-local-wedil}
    \|(i\nabla+A)u\|^2_{L^2(\R^2)}\ge C_0\, \|u\|^2_{L^2(\ball_3)}  
\end{equation}
This in combination  with the diamagnetic inequality \cite[Theorem~7.21]{lieb-loss}
\begin{equation} \label{kato}
|\nabla |u| | \leq  | (i\nabla +A)\, u| \qquad \text{a.~e.}
\end{equation}
and Lemma \ref{lem-equiv-2} implies
\begin{equation} \label{hardy-local}
   \|(i\nabla+A)u\|^2_{L^2(\R^2)}\  \gtrsim\ \| |u| \|_{H^1(\ball_3)}^2 \ \gtrsim\  \int_{\ball_3} \dfrac{|u|^2}{{r^{2}(1+\log^{2}r)}}  dx\, .
\end{equation}
It remains to extend  inequality \eqref{hardy-local} to \eqref{hardy-eq}.  We have to show that
\begin{equation} \label{infinity}
\int_{\R^2} |(i\nabla + A)\, u |^2\, dx \ \geq \  C \int_{\ball_3^c} \frac{|u|^2}{r^2 (\log r)^2}\, \, dx, 
\end{equation}
for some $C$. We follow the argument of \cite{k-jst-2011}. In view of the diamagnetic inequality \eqref{kato}, and the fact that $e<3$,  it suffices to prove  that  for all $f\in C_0^\infty(\R_+)$,
\begin{equation} \label{enough-h}
\int_0^\infty |f'(r) |^2\, r\, dr + c \int_0^e |f(r)|^2\, r\, dr \geq \, C \int_e^\infty \frac{|f(r)|^2}{r\, (\log r)^2}\,  dr .
\end{equation}
 Define the function $\xi:\R_+\to\R$ by
$$
\xi(r) = \left\{
\begin{array}{l@{\quad}cr}
c & \text{if \,} & 0< r \leq 1\\
c(2-r)  & \text{if \,} & 1 < r \leq 2
\end{array}
\right.  ,
\qquad 
\xi(r) = \left\{
\begin{array}{l@{\quad}cr}
c(r-2) &  \, \text{if }  & 2 < r \leq e \\
c (e-2)       & \,  \text{if }  & e < r  
\end{array}
\right. .
$$
An integration by parts gives
$$
\int_0^\infty |(\xi f)'(r) |^2\, r\, dr +c(1-c) \int_0^e |f(r)|^2\, r\, dr \leq \int_0^\infty |f'(r) |^2\, r\, dr + c \int_0^e |f(r)|^2\, r\, dr.
$$
Now, since $\xi(2)=0$, we integrate by parts again to get 
$$
\int_2^\infty  \left( (\xi f)'(r) -\frac{(\xi f)(r)}{2r \log r} \right)^2 r\, dr = \int_2^\infty |(\xi f)'(r) |^2\, r\, dr - \int_2^\infty \frac{|(\xi f)(r)|^2}{4 r\, (\log r)^2}\, dr.
$$
The last two equations imply \eqref{enough-h} and hence \eqref{infinity}. 
\end{proof}

To show that the power of the logarithmic factor  on the right hand side of \eqref{hardy-eq} is optimal at infinity,  we need the following result.

\begin{proposition} \label{prop-log-global}
Let $B$ be radial and assume \eqref{ass-regular}. Suppose that 
\begin{equation}  \label{ass-B0}
|B(x)| = \mathcal{O}(|x|^{-\nu}), \ \quad \nu>2, \quad |x|\to\infty . 
\end{equation} 
Assume moreover  that 
\begin{equation} \label{flux} 
\int_{\R^2} B\, dx = 2\pi m ,\qquad  m \in \Z.
\end{equation}
Suppose that there exists a function $\w:\R^2\to [0,\infty)$ such that
\begin{equation} \label{hypot}
Q_A[u]\, \geq \,  \int_{\R^2} \w\,  |u|^2\, \, dx \qquad \forall\, u\in d_0(Q_A).
\end{equation}
Then for {\rm any} $0\leq \alpha<1$ we have cf.~\eqref{log+},
\begin{equation} \label{log-weight}
\int_{\R^2} \w\, (\log_+r)^\alpha\, \, dx < \infty .
\end{equation}
\end{proposition}

\begin{proof} 
To prove \eqref{log-weight} we  write $A$ in the Poincar\'e gauge
\begin{equation} \label{poincare}
A(x) = (-x_2, x_1) \int_0^1 B(tx)\,  t\, dt  = \frac{ (-x_2, x_1) }{r^2} \int_0^r B(t)\, t\, dt.
\end{equation}
In polar coordinates we then have 
 \begin{equation}  \label{polar-coord}
Q_A[u]  =  \int_0^\infty\! \!\!\int_0^{2\pi} \Big[ \, |\partial_r u|^2 +r^{-2} |(i\partial_\theta + r a(r)) u|^2 \Big]\, r\, dr d\theta,
\end{equation}
where
\begin{equation} \label{azi}
a(r) = \frac 1r \int_0^r B(t)\, t\, dt
\end{equation}
denotes the azimuthal component of the vector potential.  Notice that by \eqref{ass-regular} and \eqref{poincare} we have $A\in L^q_{\rm loc}(\R^2)$ for some $q>2$.
Let us now construct a sequence of  test functions $u_n$ as follows;
\begin{equation} \label{test-f}
u_n(r,\theta) = e^{ i m\theta}\, (\log_+ r)^{\frac\alpha 2} \ \dfrac{\log_+(n/r)}{\log n}\, .
\end{equation}
Then $u_n \in H^1(\R^2)$ and the support of $u_n$ is compact. Hence, by Theorem \ref{thm-weidl}(1),  $u_n \in d(Q_A)$. Moreover,
\begin{equation}\label{Qn}
\limsup_{n\to\infty} Q_A[u_n] < \infty.
\end{equation}
Indeed, the radial derivative part of $Q_A[u_n]$ gives
\begin{align*}
|\partial_r u_n(x)|^2& = \dfrac{1}{(\log n)^2} \left| \partial_r \left[ (\log r)^{\frac\alpha 2} (\log(n/r)) \right] \right|^2\\[2pt]
& \le \dfrac{2}{(\log n)^2} \left( (\partial_r (\log r)^{\frac\alpha 2})^2 (\log(n/r))^2 + (\log r)^{\alpha} (\partial_r \log(n/r))^2 \right)\\[2pt]
& \leq \dfrac{2}{(\log n)^2\, r^2} \left[ (\log n)^2 (\log r)^{\alpha-2} + (\log r)^{\alpha} \right],
\end{align*}
see \eqref{test-f}. 
Hence, keeping in mind that $\alpha <1$ we obtain
\begin{align*}
\limsup_{n\to\infty}\int_0^\infty  |\partial_r u_n|^2\, r dr&=  \limsup_{n\to\infty} \int_1^n |\partial_r u_n|^2\, r dr  < \infty.
\end{align*}
As for the azimuthal derivative, we compute 
$$
 |(i\partial_\theta + r a(r)) u_n|^2 = (m- \beta(r))^2\,  |u_n|^2\, .
$$
By equation \eqref{flux} and assumptions on $B$,
\begin{align*} 
\big |m- \beta(r) \big | & \leq  \frac{1}{2\pi}\int_r^\infty\!\! \int_0^{2\pi}  |  B(t,\theta) |  \, t\, dt d\theta \, \leq \, \frac{C}{(1+r)^{\nu-2}}\, 
\end{align*}
with a constant $C$ independent of $n$. Since $\alpha<1$, this gives
\begin{align*}
\limsup_{n\to\infty}\int_0^\infty  r^{-2}  |(i\partial_\theta + r a(r)) u_n|^2\, r dr  &\leq C^2\, 
\limsup_{n\to\infty}\int_1^n  \frac{(\log r)^\alpha}{r (1+|\log r|)^2}\,  dr  < \infty\, .
\end{align*}
Hence we have proved \eqref{Qn}.
On the other hand, for any $r$ and $\theta$ we have $|u_n(r,\theta)| \leq |u_{n+1}(r,\theta)|$. Moreover, 
\begin{equation} 
\lim_{n\to\infty} |u_n(r,\theta)|^2 =  (\log_+ r)^\alpha 
\end{equation}
pointwise in $\R_+\times\Sph$. Equation \eqref{log-weight} now follows from monotone convergence by applying
inequality \eqref{hypot} to $u=u_n$ and letting $n\to\infty$.
\end{proof}

\section{\bf Singular magnetic fields}
\label{sec-singular}

\subsection{Proof of Theorem \ref{thm-singular-local}} 
\begin{proof}
As in \eqref{poincare} we use  the Poincar\'e gauge 
$$
A(x) = (-x_2, x_1) \int_0^1 B(tx)\,  t\, dt. 
$$
The quadratic form \eqref{eq-form} is then closed on the domain \eqref{domain-s}. As in the proof of Proposition \ref{prop-log-global} we
write $Q_A[u]$  in the polar coordinates;
$$
Q_A[u]  =  \int_0^\infty\! \!\!\int_0^{2\pi} \Big[ \, |\partial_r u|^2 +r^{-2} |(i\partial_\theta + r a(r,\theta)) u|^2 \Big]\, r\, dr d\theta,
$$
where
\begin{equation} \label{azi-2}
a(r,\theta) = \frac 1r \int_0^r B(t,\theta)\, t\, dt\, .
\end{equation}
 Following \cite{lw-1999} we introduce, for any fixed $r>0$, the operator $K_r$ defined on
$H^1(\Sph)$, with periodic boundary conditions, as
\begin{equation} 
K_r v = i\partial_\theta v +r a v.
\end{equation} 
In view of \eqref{azi-2} and the fact that $B\in L^1_{\rm loc}(\R^2)$ it follows that $a(r, \cdot) \in L^1(\Sph)$ for any $r>0$. 
The compactness of the imbedding $H^1(\Sph)\hookrightarrow  L^\infty(\Sph)$ then implies that for any $\eps>0$ there exists $C_\eps$ such that 
\begin{equation} 
\|a v\|_{L^2(\Sph)} \leq \|a\|_{L^1(\Sph)} \big( \eps \|v\|_{H^1(\Sph)} +C_\eps \|v\|_{L^2(\Sph)} \big).
\end{equation}
Hence the operator $ra$ is relatively bounded with respect to $i\partial_\theta$ with relative bound zero. 
By the Kato-Rellich Theorem the operator $K_r$ is thus
self-adjoint in $L^2(\Sph)$.  The eigenvalues of $K_r$ are given by 
\begin{equation} \label{lambda-m}
\lambda_m(r)  = m- \beta(r), \qquad m\in\Z.
\end{equation}
Recall that $\beta$ is defined in  \eqref{alpfa-r}. From the Stokes theorem we infer that 
$$
 \beta(r) = \frac{r}{2\pi} \int_0^{2\pi} a(r,\theta) d\theta\, .
$$
Moreover, since the spectrum of $K_r$ is purely discrete, its normalized eigenfunctions, given by
\begin{equation} \label{ef-K} 
v_m(r,\theta)= \frac{1}{\sqrt{2\pi}}\ e^{-i \big( \theta \lambda_m(r) -r\!\!\int_0^\theta a(r,s) ds\big)}\, , \qquad m\in\Z
\end{equation} 
form a complete orthonormal basis of $L^2(\Sph)$. For a given $u\in C_0^\infty(\R^2\setminus\{0\})$, we can thus expand $u(r, \cdot)$ into the Fourier series with respect to this basis to 
obtain
\begin{equation} \label{F-series}
u(r,\theta) = \sum_{m\in\Z} u_m(r) v_m(r,\theta), \qquad  u_m(r) = \int_0^{2\pi} u(r,\theta)\, \overline{v}_m(r,\theta)\, d\theta\, .
\end{equation}
The Parseval identity then gives 
\begin{equation} \label{parseval-u} 
 \int_0^{2\pi} |u(r,\theta)|^2\, d\theta = \sum_{m\in\Z} |u_m(r)|^2  \qquad a. e.  \ r>0.
\end{equation} 
 Hence, by \eqref{polar-coord},
\begin{equation} \label{parseval-Q} 
 Q_A[u] =  \int_0^\infty \!\!\int_0^{2\pi} |\partial_r u|^2\, r\, dr d\theta +  \int_0^\infty \sum_{m\in\Z}   \lambda^2_m(r)\,  |u_m(r)|^2 \, \frac{dr}{r}\, .
\end{equation} 
Since $B\in L^1_{\rm loc}(\R^2)$ we have $ \beta(r) \to 0$ as $r\to 0$. So there exists 
$\eta>0$ such that 
\begin{equation*} 
| \beta(r)| \, \leq\, \frac 14 \qquad \forall\, r \leq 2\eta.
\end{equation*}
This implies 
\begin{equation} \label{lambda-lowerb}
\lambda^2_m(r) \, \geq\,  \frac 12 \big( m^2 +\beta^2(r)\big)  \qquad  0<r\leq 2\eta , \ \  m\in\Z\, .
\end{equation}
Equation \eqref{parseval-Q} thus gives
$$
Q_A[u] \, \geq\, \frac 12 \int_{\ball_\eta} \Phi^2\, |u|^2\, dx\, . \\[2pt]
$$
Since $B_r\in L^p_{\rm loc}(\R^2)$ with $p>1$,  equation \eqref{bs-diverge} yields
$$
\liminf_{r\to 0} \Phi^2(r) >0.
$$
Hence taking $\eta$ small enough we ensure that there exists $c>0$ such that 
$$
Q_A[u] \, \geq c \int_{\ball_\eta} |u|^2\, dx.
$$
We can thus mimic the proof of Theorem \ref{thm-hardy}.  
From the diamagnetic inequality \eqref{kato} and Lemma \ref{lem-equiv-2}, applied with $\Omega=\ball_\eta$, we deduce, as in \eqref{hardy-local}, 
that for $\eta$ small enough,
$$
Q_A[u] \, \gtrsim \int_{\ball_\eta} \rho_0\, |u|^2\, dx\, . 
$$
Summing up the two bounds and extending the integral weight to all of $\R^2$ as in the proof of Theorem \ref{thm-hardy} completes the proof. 
\end{proof}

\subsection{Radial singular field} 
In this section we give a characterization of the domain of $Q_A$ in the case in which the singular part of the filed is radial. 
Following the decomposition \eqref{B-decomp} we split the vector potential accordingly as follows; 
\begin{equation} \label{A-decomp}
 A=  A_s + A_r, 
\end{equation}
where  $A_r$ is given by  \eqref{h-B}-\eqref{A-def} with $B$ replaced by $B_r$.
By assumptions on $B_r$,  there exits $q>2$ such that
\begin{equation} \label{Ar-q}
A_r\in L^q_{\rm loc}(\R^2) \, .
\end{equation}
For the singular part of the field we use again the Poincar\'e gauge. Since $B_s$ is radial, we get
\begin{equation} \label{As-definition}
 A_{s}(x) =  (-x_2, x_1) \int_0^1 B_{s}(tx)\,  t\, dt \, =  \frac{ (-x_2, x_1) }{r^2} \int_0^r B_s(t)\, t\, dt.
\end{equation}
Then
\begin{equation} \label{as-norm} 
|A_s(r)| = \frac{|\beta_s(r)|}{r},
\end{equation}
where
\begin{equation} \label{alpfa-r-s} 
\beta_s(r) = \int_0^r B_s(t)\,  t\, dt. 
\end{equation} 
The assumption  $B_s\in L^1_{\rm loc}(\R^2)$ implies that $\beta_s(r) \to 0$ as $r\to 0$. Hence there exists 
$\delta>0$ such that 
\begin{equation} \label{delta-cond} 
|\beta_s(r)| \, \leq\, \frac 14 \qquad \forall\, r \leq 2\delta.
\end{equation}
Since  $B_r\in L^p_{\rm loc}(\R^2)$ with $p>1$, in view of \eqref{bs-diverge} and \eqref{as-norm}  we can ensure that, by taking $\delta$ small enough,
\begin{equation}  \label{fluxes}
\frac{ \Phi^2(r)}{2} \, \leq  |A_s(r)|^2 \, \leq 2 \Phi^2(r) \qquad  \forall\, r \leq \delta.
\end{equation}
Let $\chi: \R^2\to [0,1]$ be a  smooth radial function such that $\chi=1$ in $\ball_\delta$ and $\chi=0$ in $\R^2\setminus \ball_{2\delta}$.

\begin{lemma} \label{lem-aux} 
Let $A$ be given by \eqref{A-decomp}, and let $\chi$ be as above. Assume that $B_s$ is radial. 
Then, for
any $u\in C_0^\infty(\R^2\setminus\{0\})$,
\begin{equation}  \label{chi-2-sided}
 \frac 12 \| \nabla (\chi u)\|_2^2 + \frac 12 \| A_s \chi u\|_2^2 \ \leq \ \| (i\nabla +A_s) (\chi u)\|_2^2 .
 \end{equation}.
\end{lemma}

\begin{proof} 
In polar coordinates,
\begin{equation} \label{pars-1}
\| (i\nabla +A_s) (\chi u)\|_2^2 = \int_0^\infty\!\!\int_0^{2\pi} \Big( |\partial_r (\chi u)|^2 + r^{-2} |(i\partial_\theta +\beta_s(r)) \chi u|^2\Big) \, r\, dr d\theta,
\end{equation}
We expand $\chi u$ into the Fourier series with respect to the standard orthonornal basis of  $L^2(\Sph)$. Thus,
\begin{equation}  \label{u-m}
(\chi u)_m(r) = \frac{\chi(r)}{\sqrt{2\pi}} \int_0^{2\pi}\!\! e^{-im\theta}\, u(r,\theta)\, d\theta = \chi(r) u_m(r), 
\end{equation}
and, in view of \eqref{pars-1},
$$
\| (i\nabla +A_s) (\chi u)\|_2^2 =  \int_0^\infty \!\!\int_0^{2\pi} |\partial_r (\chi u)|^2\, r\, dr d\theta + \int_0^\infty \sum_{m\in\Z} \,  r^{-2} (m +\beta_s(r))^2 |\chi u_m|^2 \, r\, dr.
$$
By \eqref{delta-cond} we have $\beta_s(r) \leq \frac 14$ on the support of $\chi$. Hence 
$$
 (m +\beta_s(r))^2\,  |\chi u_m|^2 \, \geq\,  \frac 12 \big( m^2 +\beta_s^2(r)\big) |\chi u_m|^2 .
$$
On the other hand, Parseval identity gives
$$
\| \nabla (\chi u)\|_2^2 = \int_0^\infty \!\!\int_0^{2\pi} |\partial_r (\chi u)|^2\, r\, dr d\theta+  \int_0^\infty \sum_{m\in\Z}\, m^2\, |\chi u_m|^2\, \frac{dr}{r} \, .
$$
This implies 
$$
\| (i\nabla +A_s) (\chi u)\|_2^2  \, \geq \, \frac 12  \| \nabla (\chi u)\|_2^2 +  \frac 12 \| \beta_s\, |x|^{-1}\, \chi u\|_2^2 
$$
and the claim follows from \eqref{as-norm}. 
\end{proof}

\begin{lemma} \label{lem-two-sided} 
Let $B$ satisfy Assumption \ref{ass-singular} and let $A$ be as above. Assume moreover that $B_s$ is radial. Let $\chi$ be as in Lemma \ref{lem-aux}. Then for any $u\in C_0^\infty(\R^2 \setminus\{0\})$ 
\begin{equation} \label{equiv-singular}
\vvvert{u}^2_A \ \asymp   \  \| (i\nabla +A_r) u\|^2_2  + \|  \chi	 \Phi  u\|_2^2 +  \|u\|^2_2\, .
\end{equation} 
\end{lemma}

Recall that the $\vvvert{u}^2_A $ is defined in \eqref{vvv}.

\begin{proof}
Recall that $B_s$ is compactly supported.
Below we denote by $C$ a generic positive number whose value may change
line to line. Owing to \eqref{fluxes} it suffices to prove that 
\begin{equation} \label{enough}
\vvvert{u}^2_A \ \asymp   \  \| (i\nabla +A_r) u\|^2_2  + \|  \chi A_s u\|_2^2 +  \|u\|^2_2\, .
\end{equation} 
Given $u\in C_0^\infty(\R^2\setminus\{0\})$, we have 
\begin{equation} \label{upper}
Q[u] = \| (i\nabla +A) u\|_2^2  \leq 2 \| (i\nabla +A_r) u\|_2^2 + 2\| A_s u\|_2^2 \, . 
\end{equation}
To prove the opposite inequality we use a type of the IMS localization formula as it was done in the proof of \cite[Prop.~1.1]{cf}, 
see also \cite[Sec.~2.1]{bcf}. Let $\chi$ be as in Lemma \ref{lem-aux} and  let $\zeta: \R^2\to [0,1]$
satisfy $\chi^2+\zeta^2 =1$. Using the fact that $\chi \nabla \chi +\zeta\nabla \zeta=0$ we get
\begin{equation} \label{ims}
| (i\nabla +A) u|^2 = | (i\nabla +A) (\chi u)|^2+| (i\nabla +A) (\zeta u)|^2 - |u\nabla \chi|^2 -|u \nabla \zeta|^2\, .
\end{equation}
The elementary inequality 
\begin{equation} \label{eta}
|a+b|^2 \geq (1-\eta) |a|^2 +(1-\eta^{-1}) |b|^2 \qquad \forall\,  a,b\in\R^2, \ \eta\in (0,1)
\end{equation}
now gives 
\begin{align*}
 | (i\nabla +A) (\chi u)|^2 + | (i\nabla +A) (\zeta u)|^2 & \geq (1-\eta)  | (i\nabla +A_s) (\chi u)|^2 +(1-\eta^{-1}) |A_r (\chi u)|^2\\[4pt]
 & \ +  (1-\eta)  | (i\nabla +A_r) (\zeta u)|^2 +(1-\eta^{-1}) |A_s (\zeta u)|^2
\end{align*}
By  assumptions on $B_s$ we have $|\zeta A_s|\in L^\infty(\R^2)$. Since $\chi$ and $\zeta$ are smooth, it follows from \eqref{ims} that
\begin{align} 
\| (i\nabla +A) u\|_2^2 & \geq   (1-\eta)  \| (i\nabla +A_s) (\chi u)\|_2^2 +  (1-\eta)  \| (i\nabla +A_r) (\zeta u)\|^2  \nonumber \\[4pt]
& \quad +(1-\eta^{-1}) \| \chi A_r  u\|^2  - C \|u\|_2^2\label{lowerb-aux}
\end{align}
Using  Lemma \ref{lem-aux} and equation \eqref{eta} a couple of times we deduce that 
\begin{align*}
 \| (i\nabla +A_s) (\chi u)\|_2^2 \, & \geq\, \frac 12 \| \nabla (\chi u)\|_2^2 + \frac 12 \| A_s \chi u\|_2^2\, \geq\,  
 \frac{1-\eta}{2}\,  \|\chi  \nabla  u\|_2^2 + \frac 12 \| A_s \chi u\|_2^2 - C  \|u\|_2^2\\[4pt] 
 & \geq  \frac{(1-\eta)^2}{2}\,  \|\chi (i \nabla+A_r)  u\|_2^2 + \frac 12 \| A_s \chi u\|_2^2  +  \frac{1-\eta^{-1}}{2} \| \chi A_r  u\|^2- C  \|u\|_2^2 \\[4pt]
  & \geq  \frac{(1-\eta)^2}{2}\,  \|\chi (i \nabla+A_r)  u\|_2^2 + \frac{1-\eta}{2} \| A_s  u\|_2^2  + \frac{1-\eta^{-1}}{2} \| \chi A_r  u\|^2- C  \|u\|_2^2.
\end{align*}
Here, in the last line, we have used again the fact that   $|\zeta A_s|\in L^\infty(\R^2)$, which implies
$
 \| A_s \chi u\|_2^2\, \geq \,  (1-\eta) \| A_s u\|_2^2 -C \|u\|_2^2\ . 
$
Similarly, we estimate 
$$
\| (i\nabla +A_r) (\zeta u)\|^2 \, \geq (1-\eta) \| \zeta (i\nabla +A_r) (\zeta u)\|_2^2 - C\|u\|_2^2\, 
$$
Hence after inserting the above estimates into \eqref{lowerb-aux} and setting $\eta= \frac 12$ we obtain 
\begin{align}\label{aux-2}
 \| (i\nabla +A) u\|_2^2  & \geq\, 
\frac{1}{16}     \| (i\nabla +A_r) u\|_2^2 +\frac 18 \| A_s u\|_2^2 - C  \|u\|_2^2 - C \| \chi A_r  u\|_2^2 \, .
\end{align}
To control the last term we notice that, 
the compactness of the imbedding  $H^1(\ball_{2\delta}) \hookrightarrow L^s(\ball_{2\delta})$ for any  $s > 2$ assures 
that for any $\eps>0$ there exists a constant $C_\eps$ such that 
$$
\| \chi u\|^2_{L^s(\ball_R)} \leq \eps  \| \nabla ( \chi u)\|_2^2 + C_\eps  \|\chi u\|_2^2
$$
Since $| A_r| \in L^q_{\rm loc}(\R^2)$ with $q>2$, see Lemma \ref{lem-A-regular}, this  combined with the H\"older inequality 
and Theorem \ref{thm-weidl} gives
\begin{align*} 
\| \chi A_r  u\|_2^2 & \leq\,  \eps\,  \|A_r\|_{L^q(\ball_{2\delta})}\,  \| \nabla  (\chi u)\|_2^2 + C_\eps  \|u\|_2^2  
\leq \eps\, C   \| (i\nabla +A_r) (\chi u)\|_2^2 + \widetilde C_\eps   \|u\|_2^2  \\[4pt]
& \leq  2 \eps\, C   \| (i\nabla +A_r)  u\|_2^2 + 2 \eps\, C \|\nabla \chi\|^2_\infty  \|u\|_2^2+ \widetilde C_\eps   \|u\|_2^2\, .
\end{align*}
Taking $\eps$  small enough and inserting this bound  into \eqref{aux-2} we finally obtain 
\begin{align}
 \| (i\nabla +A) u\|_2^2 +  \|u\|_2^2  \   \gtrsim \    \| (i\nabla +A_r)  u\|_2^2+   \|u\|_2^2 + \| A_s u\|_2^2 .
\end{align}
This completes the proof of \eqref{enough} and hence of the lemma.

\end{proof}

\begin{corollary} \label{cor-domain} 
Let $B$  satisfy Assumption \ref{ass-singular} with $B_s$ being radial. Let $A$ be given by \eqref{A-decomp}. 
Then the form domain  \eqref{domain-s} satisfies 
\begin{equation} \label{f-domain}
d(Q_A) = \big\{ u\in d_0(Q_{A_r}) : \  | \Phi u| \in L_{\rm loc}^2(\R^2)\big\}\, ,
\end{equation}
Moreover, if
\begin{equation} \label{weak} 
\limsup_{r\to 0} |\, \beta(r) \log r| < \infty , 
\end{equation}
then $d(Q_A)=d_0(Q_{A_r})$.

\end{corollary}

\begin{proof} 
Since $\Phi$ is bounded on compact subsets of $\R^2\setminus\{0\}$, see \eqref{alpfa-r} and \eqref{Phi}, $ | \Phi u| \in L_{\rm loc}^2(\R^2)$ if and only 
if  $|\chi \Phi  u| \in L^2(\R^2)$. 
Now the first claim follows from Lemma \ref{lem-two-sided} and the fact that, in view of Theorem \ref{thm-weidl}(1), 
\begin{equation} \label{no-zero}
\overline{C_0^\infty(\R^2\setminus\{0\})}^{\, \vvvert{ \cdot}_{A_r}}\, =\, \overline{C_0^\infty(\R^2)}^{\, \vvvert{ \cdot}_{A_r}}\, .
\end{equation}
To prove the second claim note that if \eqref{weak} is satisfied, then by Lemma \ref{lem-equiv-2} and  Theorem \ref{thm-weidl}(1)  we have
$$
\|\chi \Phi  u\|_2^2 \, \lesssim\,  \| (i\nabla +A_r)  u\|_2^2 + \| u\|_2^2 .
$$
Lemma  \ref{lem-two-sided} then implies
$
\vvvert{u}_A \, \asymp \vvvert{u}_{A_r} \, ,
$
and the claim follows from \eqref{no-zero}.
\end{proof}

\subsection{Examples} We conclude this section with a couple of examples.
\label{sec-examples}

\subsubsection*{Example 1}
Consider the singular magnetic field  
\begin{equation} \label{ex-1}
B_s(r) = 
 \left\{
\begin{array}{l@{\quad}cr}
\displaystyle{\frac{b_0}{r^2  |\log r|^{\gamma}}} &  \, \text{if }  & r \leq e^{-1} \\[10pt]
0        & \,  \text{if}  & e^{-1} <r
\end{array}
\right. , \qquad \gamma >1
\end{equation}
where  $b_0\neq 0$ is a constant. The condition $\gamma>1$ assures that $B_s\in L^1_{\rm loc}(\R^2)$. 
A short calculation shows that 
$$
 \int_{\ball_r} B_s\, dx = \frac{2\pi b_0}{\gamma-1} \, |\log r|^{1-\gamma} \qquad r \leq e^{-1}\, .
$$
Condition \eqref{bs-diverge} is thus satisfied and Theorem \ref{thm-singular-local} gives 
$$
\rho(r) \, = \,  r^{-2} \Big ( |\log r|^{-2} + \frac{b_0^2}{(\gamma-1)^2}  |\log r|^{ 2- 2\gamma}\Big), \qquad   r\leq e^{-1}
$$ 
Hence $\rho \asymp \rho_0$ as long as $\gamma\geq 2$. In fact, by Corollary \ref{cor-domain}, 
we have $d(Q_A)=d_0(Q_{A_r})$ for all $\gamma\geq 2$.


\subsubsection*{Example 2}
Let  
\begin{equation} \label{ex-2}
B_s(r) = 
 \left\{
\begin{array}{l@{\quad}cr}
\displaystyle{\frac{b_0}{r^2\, |\log r|  \big(\log |\log r|\big) ^{\gamma}}} &  \, \text{if }  & r \leq e^{-2} \\[10pt]
0        & \,  \text{if}  & e^{-2} <r
\end{array}
\right. , \qquad \gamma >1, \  \ b_0\neq 0. 
\end{equation}
Then we have again $B_s\in L^1_{\rm loc}(\R^2)$ and Theorem \ref{thm-singular-local} implies
$$
\rho(r) = \frac{b_0^2}{(\gamma-1)^2\, r^2} \big(\log |\log r|\big)^{2- 2\gamma} \, \big(1+o(1)\big)   \qquad r\to 0.
$$

\section{\bf Eigenvalues of magnetic Schr\"odinger operators }
\label{sec-application}
In this section we consider an application of inequality \eqref{hardy-eq} to spectral estimates for magnetic 
 Schr\"odinger operators of the type
\begin{equation} \label{hamiltonian} 
(i\nabla +A)^2 -V
\end{equation} 
in $L^2(\R^2)$. Here $V:\R^2\to\R$ is an electric potential. We denote by $N((i\nabla +A)^2 -V)$ 
 the number of negative eigenvalues, counted with their multiplicities, of the operator \eqref{hamiltonian}.  
 To motivate our result let us recall some known upper bounds on $N((i\nabla +A)^2 -V)$ established earlier in the 
literature. Set
 \begin{equation} \label{V-radial}
 \V(r) := {\rm ess} \sup_{0\leq\theta\leq 2\pi} V_+(r,\theta)
 \end{equation}
 and let 
 \begin{equation} 
  L^1(\R_+, L^\infty(\Sph)) = \Big\{V\in L^1_{\rm loc}(\R^2):  \ \int_0^\infty  \V(r)\, r\,dr <\infty \Big\}.
 \end{equation}

In \cite{k-jst-2011, BK} it has been shown, under certain conditions on $B$, that
\begin{equation} \label{jst-2011}
N((i\nabla +A)^2 -V) \ \lesssim\ \int_0^\infty \V(r) (1+|\log r|)\, r\, dr, 
\end{equation}
and that for any $\alpha>0$ there exists a constant $c_\alpha$ such that
\begin{equation} \label{jst-2011-b}
N((i\nabla +A)^2 -V) \ \leq\, c_\alpha  \int_{\R^2}\, V_+^{1+\alpha} (1+|x|)^{2\alpha}\, dx\, .
\end{equation}
These bounds hold for all $V$ for which the corresponding right hand side is finite. Similar bounds have been derived in \cite{FLR} for the Aharonov-Bohm magnetic field.

However, neither \eqref{jst-2011} nor \eqref{jst-2011-b} can be applied to potentials with singularities strong enough as 
to produce a  non semi-classical behavior of the counting function 
$N((i\nabla +A)^2 -V)$ in the strong coupling limit. A typical example of such a potential is 
\begin{equation} \label{local}
V_\sigma(x) = \left\{
\begin{array}{l@{\quad}cr}
r^{-2}\, |\log r|^{-2}\, \big(\log |\log r|\big)^{-1/\sigma} & \text{if\, \,} & r<
e^{-2}  \\[3pt]
0 & \, \, \text{if \, }  & r \geq e^{-2} 
\end{array}
\right. 
\end{equation}
with $\sigma >1$, which has been invented in \cite{BS} for two-dimensional Schr\"odinger operators without magnetic field. Note 
that for any $\sigma>1$ we have $V_\sigma\in L^1(\R^2)$ while the right hand sides of  \eqref{jst-2011} and \eqref{jst-2011-b} are infinite for $V=V_\sigma$.  
In fact, using the results of  \cite{BS} 
it was proved in \cite{k-jst-2011}, \cite{BK} that, under suitable conditions on $B$, 
\begin{equation} \label{strong}
N\big( (i\nabla +A)^2 -   \lambda V_\sigma \big) \, \asymp\, \lambda^\sigma  \qquad \text{as} \quad \lambda\to\infty\, .
\end{equation}

With the help of inequality \eqref{hardy-eq} we shall prove an upper bound on $N((i\nabla +A)^2 -V)$ which is applicable also to singular potentials similar to
 \eqref{local}. We will also show that our upper bound is order sharp in the strong coupling limit, see Remark \ref{rem-singular}.
 
To state our result we need some additional notation. 
Given $a>1$ and $V\in  L^1(\R_+, L^\infty(\Sph))$ we define the functional 
\begin{equation} \label{V-ap}
[ V]_a  =\Big(  \sup_{t>0}\, t^{a} \int_{\Omega_t} \rho_0(r)\, (1+|\log r|)\,  dx \Big)^{\frac{1}{a}}\, ,
\end{equation}
where $\rho_0$ is given by \eqref{weight-rho} and where 
$$
\Omega_t = \big\{ x\in\R^2:  \V(r)   >t \rho_0(r) \big\} \, .
$$
We have

\begin{theorem} \label{thm-counting}
Let $B$ satisfy  \eqref{ass-B0} and assume that $B\in L^p_{\rm loc}(\R^2)$ for some $p>2$. Then for any $a>1$ there exists  a constant $c(B,a)$ such that 
\begin{equation}  \label{counting}
N((i\nabla +A)^2 -V) \, \leq\,   c(B,a) \, [ V]_a ^a \\[2pt]
\end{equation}
for all  $V \in L^1(\R_+, L^\infty(\Sph))$ for which the right hand side is finite.  
\end{theorem}

\begin{proof} 
The proof is a straightforward modification of \cite[Theorem~6.1]{BK} which uses the method of \cite{FLR}, see also \cite{BS}. 
Let $C_B$ be the constant in \eqref{hardy-eq} and let $0<\delta<1$. We set $s= \delta\, C_B$ and using \eqref{hardy-eq} we 
estimate $(i\nabla +A)^2 -\V$
in the sense of quadratic forms in the following way:
\begin{align} 
(i\nabla +A)^2-  \V &= \delta ((i\nabla +A)^2-  \rho_0\,  C_B ) +(1-\delta) \big( (i\nabla +A)^2-  (1-\delta)^{-1} \rho_0 (\rho_0^{-1} \V-s)\big) \nonumber \\[4pt]
 & \geq (1-\delta) \big( (i\nabla +A)^2-  (1-\delta)^{-1} \rho_0 (\rho_0^{-1} \V-s)_+\big) . \label{eta-1}
\end{align}
Hence
$$
N((i\nabla +A)^2 -V) \leq N((i\nabla +A)^2 -\V)\, \leq \, N((i\nabla +A)^2  -  (1-\delta)^{-1} \rho_0 (\rho_0^{-1} \V-s)_+).
$$
From \cite[Corollary~1.3]{BK} and the hypotheses on $B$ it follows that estimates \eqref{jst-2011} holds. Applying the latter to the last term we obtain
$$
N((i\nabla +A)^2 -V) \, \leq\,  C (1-\delta)^{-1}  \int_0^\infty \rho_0(r) \big(\rho_0(r)^{-1} \V(r)-s\big)_+\,  (1+|\log r|)\, r \, dr.
$$
The layer cake representation, \cite[Theorem~1.13]{lieb-loss},  gives
\begin{equation} 
 \big(\rho_0(r)^{-1} \V(r)-s\big)_+ =  \int_{\{ \V(r)>\rho_0(r)(s+z)\}}\, dz\, .
\end{equation}
Hence
\begin{align*}
\int_0^\infty\!\! \rho_0(r)\, (\rho_0(r)^{-1} \V(r)-s\big)_+  (1+ |\log r|)\, r \, dr & = \int_0^\infty\!\! \int_{\{ \V(r)>\rho_0(r)(s+z)\}} \rho_0(r)\, (1+ |\log r|)\, r \, dr dz\\[5pt]
&\leq  \int_0^\infty  \sup_{t>0} \Big( t^a\!\! \int_{\Omega_t} \rho_0\, (1+|\log r|)\, r dr \Big)\, \frac{dz}{(s+z)^{a}} \\[5pt]
&= \frac{s^{1-a}}{a-1} \ [ V]_a ^{a}\, .
\end{align*}
This completes the proof of \eqref{counting}.
\end{proof}

\begin{remark} \label{rem-singular}
A quick calculation shows that $[V_\sigma]_a$ is finite if and only if $\sigma\leq a$. Moreover,  $[\lambda V]_a = \lambda\,  [V]_a$. Hence 
upon setting $a= \sigma$ in Theorem \ref{thm-counting} we obtain 
$$
N((i\nabla +A)^2 -\lambda  V_\sigma) \, \leq\,  \lambda^\sigma \, c(B,\sigma) \, [ V_\sigma]_\sigma^{\sigma}\, . 
$$
Notice that this bound captures the correct behavior of $N((i\nabla +A)^2 -\lambda V_\sigma)$ in the strong coupling limit, cf.~equation \eqref{strong}.
\end{remark}


\section*{\bf Acknowledgement}
The authors thank the referees for helpful comments and suggestions.

\bibliographystyle{amsalpha}

\end{document}